\documentclass[onecolumn,draftclsnofoot]{ieeetran}

\usepackage{amsmath,amsthm}
\usepackage{amssymb}
\usepackage{amsfonts}
\usepackage{amsbsy}
\usepackage{mathtools}
\usepackage{xcolor}
\usepackage{soul}
\usepackage{comment}
\usepackage{nicematrix}
\usepackage{cite}

\theoremstyle{plain}
\newtheorem{theorem}{Theorem}
\newtheorem{lemma}{Lemma}

\theoremstyle{remark}

\newtheorem{remark}{Remark}

\theoremstyle{definition}
\newtheorem{definition}{Definition}

\DeclareMathOperator{\rk}{rk}

\renewcommand{\baselinestretch}{1.2}

\definecolor{applegreen}{rgb}{0.55, 0.71, 0.0}

\title{A Matrix Completion Approach for the Construction of MDP Convolutional Codes}

\IEEEoverridecommandlockouts
\author{\large
Sakshi Dang\thanks{Sakshi Dang is with the Department of Mathematics, Indian Institute of Technology Bombay, Mumbai, India, (email: sakshidang10@gmail.com). She is supported by a research fellowship from the Council of Scientific and Industrial Research (CSIR), Government of India.}, Julia Lieb,\thanks{Julia Lieb is with the Department of Mathematics, TU Ilmenau, Ilmenau, Germany, ({email: julia.lieb@tu-ilmenau.de}). She is supported by the German research foundation, project number 513811367.} Okko Makkonen\thanks{Okko Makkonen is with the Department of Mathematics and Systems Analysis, Aalto University, Espoo, Finland, {(email: okko.makkonen@aalto.fi). He is supported by the Vilho, Yrjö and Kalle Väisälä Foundation of the Finnish Academy of Science and Letters.}}, Pedro Soto\thanks{Pedro Soto is with the Department of Mathematics, Virginia Tech, Blacksburg, VA, USA (email: pedrosoto@vt.edu). He is supported by the Amazon-Virginia Tech Initiative for Efficient and Robust Machine Learning.}, and Alex Sprintson\thanks{Alex Sprintson is with the Department of Electrical and Computer Engineering, Texas A\&M University,  TX, USA (email:  spalex@tamu.edu). His work is supported by the US NSF under Grant CIF 2327510.} \thanks{This work was made possible due to the support of  the Steger Center for International Scholarship - Virginia Tech in Riva San Vitale, Switzerland.}
}

\date{\today}

\begin{document}

\maketitle

\begin{abstract}

Maximum Distance Profile (MDP) convolutional codes are an important class of channel codes due to their maximal delay-constrained error correction capabilities. The design of MDP codes has attracted significant attention from the research community. However, only limited attention was given to addressing the complexity of encoding and decoding operations. 
This paper aims
to reduce encoding complexity by constructing partial unit-memory MDP codes with structured and sparse generator matrices. In particular, we present a matrix completion framework that extends a structured superregular matrix (\emph{e.g.,} Cauchy) over a small field to a sparse sliding generator matrix of an MDP code. We show that the proposed construction can reduce the encoding complexity compared to the current state-of-the-art MDP code designs.

\end{abstract}

\section{Introduction}

Error-correcting codes are used to protect data against transmission errors in communication systems. In contrast to the classical block codes, where data blocks of fixed length are encoded and decoded independently, convolutional codes possess memory, which creates dependencies among consecutive blocks of encoded symbols. This is achieved via using a polynomial generator matrix $G(z)\in\mathbb F_q[z]^{k\times n}$ over 
some finite field $\mathbb F_{q}$ and defining the convolutional code to be the module $\{m(z)G(z)\ |\ m(z)\in\mathbb F_{q}[z]^k\}$. 
The memory of such a convolutional code is then equal to the degree of the matrix polynomial $G(z)=\sum_{t\in\mathbb N_0}G_tz^t$ with coefficient matrices $G_t\in\mathbb F_{q}^{k\times n}$ for $t\in\mathbb N_0$. As the computational complexity increases with the memory, so-called (partial) unit-memory convolutional codes, \emph{i.e.},
codes with generator matrices of the form $G_0+G_1z$, are of particular interest \cite{pum, wachter2012efficient}. Another parameter
of a convolutional code is its degree $\delta$, defined as the maximal degree of all full-size minors of $G(z)$.

The error-correcting capability of a  code can be measured by its  minimum distance, where, depending on the application, different distance notions can be considered. For sequential encoding and decoding of convolutional codes, the so-called $j$-th \textbf{column distances} are the most relevant measure. In particular, for low-delay applications, it is important to maximize the \textbf{first column distances} 
 of the code (see Section~\ref{sec:pre} for definitions). There exist upper bounds for these column distances depending on the parameters of the code \cite{MDS-Conv}. The convolutional codes that reach these upper bounds  for the largest possible value of $j$
 are called \textbf{Maximum Distance Profile (MDP)}. 
A partial unit-memory convolutional code is MDP if and only if all full-size minors of the matrix
\begin{equation*}
    G_1^c \coloneqq \begin{pmatrix}
        G_0 & G_1 \\ \mathbf{0} & G_0
    \end{pmatrix} \in \mathbb{F}_q^{2k \times 2n}
\end{equation*}
that can be nonzero are indeed nonzero. Note that this implies that it is a necessary condition for the convolutional code to be MDP that $G_0$ is an MDS matrix.
In \cite{hutchinson2005convolutional} it has been shown that MDP convolutional codes exist for all code parameters if the size of the underlying finite field is sufficiently large. However, it turns out to be a challenging task to provide algebraic constructions of MDP convolutional codes over small finite fields~\cite{alfarano2024weighted,chen2023lower,almeida2013new,MDS-Conv,luo2023construction,tomas2012decoding}.

This paper aims to reduce encoding complexity by constructing
partial unit-memory MDP codes with structured and sparse
generator matrices. 
 To this end it is desirable to have MDP codes over fields of the smallest possible size. At the same time the complexity can be reduced if many of the entries of the coefficient matrices are lying in a subfield of $\mathbb F_{q}$ or are even equal to zero. A third way to optimize complexity is to use structured coefficient matrices.

Our main contribution is the construction of partial unit-memory MDP codes with $k>n-k=\delta$ where $G_0$ is a given structured MDS matrix over a field of smallest possible size $q$. In particular, we choose $G_0$ to be a  Cauchy matrix. 
We then identify a matrix $G_1$ with a maximal amount of zero entries that completes $G_1^c$ such that $G(z)$ is the generator matrix of an MDP convolutional code, at the same time choosing the nonzero entries in $G_1$ from $\mathbb F_{q^d}$ with $d\in\mathbb N$ as small as possible. 
With this procedure, we can view the construction of MDP codes as a matrix completion problem~\cite{8786158,8555105}.

The paper is structured as follows.
In Section \ref{sec:pre}, we present the necessary background on convolutional codes. In Section~\ref{sec:ex}, we give a theoretical upper bound on the field size for the existence of partial-unit memory MDP codes with $k>n-k=\delta$. In Section \ref{sec:construct}, we present our main result, which is a construction of partial unit-memory MDP codes with the above properties. We show that it is sufficient to choose $d=\lceil\frac{\delta^2-1}{4}\rceil +1$ but conjecture that our construction is MDP also for smaller values of $d$. In Section \ref{sec:comp}, we analyze the encoding complexity of our construction and show that it can reduce the encoding complexity compared to the current state-of-the-art MDP code designs.

\section{Preliminaries}\label{sec:pre}

Throughout this paper we denote by $\mathbb F_q$ the finite field with $q$ elements and by $\mathbb{F}_{q}[z]$ the ring of polynomials with coefficients in $\mathbb{F}_{q}$. We use $[n]$ to denote the set $\{1, 2, \dots, n\}$ and $i + S$ = $\{i + s \mid s \in S\}$ for a set $S$.

A linear block code of rate $k / n$ is a $k$-dimensional subspace $\mathcal{C}$ of $\mathbb{F}_q^n$. A matrix $G \in \mathbb{F}_q^{k \times n}$ is a generator matrix of $\mathcal{C}$ if the rows of $G$ span $\mathcal{C}$. The space $\mathbb{F}_q^n$ is equipped with the metric induced by the Hamming weight $wt(v)$, which is the number of nonzero coordinates in $v \in \mathbb{F}_q^n$. The minimum distance $d(\mathcal{C})$ of a linear code is the minimal 
Hamming weight of any nonzero codeword. The famous Singleton bound states that $d(\mathcal{C}) \leq n - k + 1$ and codes that achieve this bound with equality 
are called maximum distance separable (MDS) codes. We say that a matrix $G \in \mathbb{F}_q^{k \times n}$ has the MDS property if it generates an MDS code. It is well known that $G$ has the MDS property if and only if any $k \times k$ minor of $G$ is nonzero.  
We call a matrix whose minors are all non-zero a \textbf{superregular matrix}.
An example of a family of superregular matrices is Cauchy matrices. A $k \times n$ \textbf{Cauchy matrix} is a matrix $C$ whose entries are given by $C_{ij} = \frac{1}{\alpha_j - \beta_i}$ where $\alpha_1, \dots, \alpha_n, \beta_1, \dots, \beta_k \in \mathbb{F}_q$ are distinct elements.

\begin{definition}
    A \textbf{convolutional code} $\mathcal{C}$ of rate $k/n$ over $\mathbb{F}_q$ is an $\mathbb{F}_{q}[z]$-submodule of $\mathbb{F}_{q}[z]^n$ of rank $k$. A polynomial matrix $G(z)\in\mathbb{F}_{q}[z]^{k\times n}$ of full row rank such that
    \begin{eqnarray}
\mathcal{C} =  \{v(z) = u(z)G(z) \in \mathbb{F}_{q}[z]^n\ |\ u(z) \in \mathbb{F}_{q}[z]^k\} \nonumber
\end{eqnarray}
is called \textbf{generator matrix} of $\mathcal{C}$.
The maximum degree of the full-size minors of a generator matrix of $\mathcal{C}$ is called the \textbf{degree} $\delta$ of $\mathcal{C}$. A convolutional code of rate $k/n$ and degree $\delta$ is also denoted as $(n, k, \delta)$ convolutional code.
\end{definition}

For convolutional codes (in the Hamming metric), there are two main distance measures: the free distance, which measures the full error-correcting capability of the code and, for any $j\in\mathbb N_0$, the $j$-th column distance, which measures how many errors can be corrected with time delay $j$.

\begin{definition}
The \textbf{free distance} of a convolutional code $\mathcal{C}$ is
$$d_{free}(\mathcal{C}):= \min \{wt(v(z))\ |\ v(z) \in \mathcal{C}\setminus\{0\}\},$$ where $wt(v(z))\hspace{-1mm} = \hspace{-1mm}\sum^{\deg(v(z))}_{t=0} wt (v_t)$
for $v(z)\hspace{-1mm} =\hspace{-1mm} \sum_{t=0}^{\deg (v(z))}v_{t}z^{t}$.
For $j\in\mathbb N_0$, the \textbf{$j$-th column distance}\index{column distance} of $\mathcal{C}$ is defined as
\[
d_j^c(\mathcal{C}):=\min\left\{\sum_{t=0}^j wt(v_t)\ |\ {v}(z)\in\mathcal{C} \text{ and }{v}_0 \neq 0\right\}.
\]
\end{definition}

There is a Singleton-like upper bound for the free distance of a convolutional code called a \emph{generalized Singleton bound}.

\begin{theorem}\cite{rosenthal1999maximum}\label{thm:singleton}
    Let $\mathcal{C}$ be an $(n,k,\delta)$ convolutional code. Then,
    $d_{free}(\mathcal{C}) \leq (n-k)\Big( \Big\lfloor \frac{\delta}{k} \Big\rfloor +1 \Big) + \delta +1.$
\end{theorem} 

For the $j$-th column distance, one has the following upper bound.

\begin{theorem}\cite{MDS-Conv}\label{thm:column} Let $\mathcal{C}$ be an $(n,k,\delta)$ convolutional code. Then, for $j\in\mathbb N_0$,
$$d_j^c (\mathcal{C}) \leq (n-k)(j + 1) + 1.$$
Moreover, if $d^c_j = (n-k)(j+1)+1$ for some $j \in \mathbb N_0$, then $d^c_i = (n-k)(i+1)+1$, for $i \leq j$. 

\end{theorem}

It is easy to see that $d_0^c(\mathcal{C})\leq d_1^c(\mathcal{C})\leq \cdots \leq d_{free}(\mathcal{C})$ and hence, it follows from the two previous theorems
that $d_j^c (\mathcal{C}) = (n-k)(j + 1) + 1$ implies  $j\leq L$ with

\begin{equation*}
L:=\left\lfloor\frac{\delta}{k}\right\rfloor+\left\lfloor\frac{\delta}{n-k}\right\rfloor.
\end{equation*}

\begin{definition}\label{def:MDP_code}
An $(n,k,\delta)$ convolutional code $\mathcal{C}$ is called \textbf{maximum distance profile (MDP)} if $d_L^c (\mathcal{C}) = (n-k)(L + 1) + 1$.
\end{definition}

It has been shown that MDP convolutional codes exist for all $(n,k,\delta)$ and sufficiently large finite fields \cite{hutchinson2005convolutional}. Besides, several constructions have been provided, see \cite{alfarano2024weighted,chen2023lower,almeida2013new,MDS-Conv,luo2023construction,tomas2012decoding}. However, it turns out to be a difficult task to construct MDP codes over small finite fields. This can be explained with the following theorem, which provides a criterion to check whether a convolutional code is MDP. For this theorem, we need the \textbf{sliding generator matrix} 
defined by $$G_j^c:=\begin{pmatrix}
G_0 & \cdots & G_{j}\\
 & \ddots & \vdots\\
\mathbf{0} &  & G_0
\end{pmatrix}\in\mathbb F_q^{(j+1)k \times (j+1)n}\quad
\text{for}\quad j\in\mathbb N_0.
$$ 

\begin{theorem}\cite{MDS-Conv}\label{thm:MDP_criterion}
Let $\mathcal{C}$ be a convolutional code with generator matrix $G(z)=\sum_{i=0}^{\mu}G_iz^i\in \mathbb{F}_q[z]^{k\times n}$ and $0\leq j\leq L:=\left\lfloor\frac{\delta}{k}\right\rfloor+\left\lfloor\frac{\delta}{n-k}\right\rfloor$.
The following statements are equivalent:
  \begin{itemize}
  \item[(a)] $d_j^c (\mathcal{C})=(n-k)(j + 1) + 1$.
  \item[(b)]
  Every full-size minor of  $G_j^c\in\mathbb F_q^{(j+1)k \times (j+1)n}$
  formed
from the columns with indices $ 1 \leq t_1 < \cdots < t_{(j+1)k}$, where $t_{sk+1} > sn$, for $s=1,2, \dots, j$, is nonzero. Here we define $G_i=0$ for $i>\mu$.
  \end{itemize}
\end{theorem}

\begin{remark}\label{rem:non_triv}
For $i\in\{0,\hdots,j\}$, let $I_{i}$ be the set of indices chosen from columns $[in + 1,(i+1)n]$ of $G_j^c$ to form the minors in part (b) of Theorem~\ref{thm:MDP_criterion}. Then, $|\cup_{i \leq \ell} I_i| \leq (\ell+1) k$
and $|\cup_i I_i| = (j + 1)k$. 
We call minors of this form as  \textbf{non-trivial minors}.
Note that no full-size minor of $G_j^c$ not fulfilling the conditions of the previous theorem can be nonzero due to the zeros below the block diagonal in $G_j^c$. 
Hence, all full-size minors of $G_j^c$ that can be nonzero have to be nonzero, causing that the number of minors to check is increasing very rapidly when increasing the code parameters.
\end{remark}

A convolutional code with $\mu=\deg(G(z))=1$ is called \textbf{unit-memory} if $\rk(G_1)=k$ and \textbf{partial unit-memory} if $\rk(G_1)<k$. 
In the following, we want to restrict to the case of partial unit-memory convolutional codes with $k>n-k=\delta$, i.e. $L=1$ and $\rk(G_1)=n-k$; note that the case of codes with $\mu=1$ and $\delta=\min\{n-k,k\}$  has also been considered in \cite{luo2023construction,chen2023lower}.
We have to find $G_0,G_1\in\mathbb F_q^{k\times n}$ such that all full-size minors of 
\begin{equation}\label{eqn:G_1^c}
    G_1^c=\begin{pmatrix}
    G_0 & G_1\\
    \mathbf{0} & G_0
\end{pmatrix}
\end{equation}
where we choose at most $k$ columns out of the first $n$ columns of $G_1^c$ are nonzero. Note that this implies that $G_0$ has to be an MDS matrix (this can be seen by considering the case that we choose exactly $k$ out of the first $n$ columns from $G_1^c$).

\begin{remark}\label{minrank}
    A generator matrix $G(z)$ of a convolutional code with optimal first column distance fulfills $\rk (G_1)  \geq n-k$. An interesting property of our construction in Section~\ref{sec:construct} is that $G_1$ has minimum possible rank of an MDP code with $L\geq 1$. 
\end{remark}

 \begin{proof}
 Consider a minor of $G_1^c$ where we choose the last $n$ columns of this matrix as well as $2k-n$ columns out of the first $n$ columns, i.e. a minor corresponding to a submatrix of $G_1^c$ of the form $$\begin{pmatrix}
     A & G_1\\
  0 & G_0
\end{pmatrix}$$ where $A\in\mathbb F_q^{k\times (2k-n)}$ is a submatrix of $G_0$. If the corresponding convolutional code has optimal first column distance, then (according to Theorem \ref{thm:MDP_criterion}), this minor has to be nonzero.
 If $\rk (G_1)  < n-k$, then we can achieve via row operations that $$\begin{pmatrix}
    A & G_1\\
     0 & G_0
 \end{pmatrix}$$ is transformed into 
 $$\begin{pmatrix}
   \tilde{A}_1      & 0_{(2k-n+1)\times n}\\
    \tilde{A}_2 & \tilde{G}_1\\
    0 & G_0
 \end{pmatrix}$$ with $\tilde{A}_1\in\mathbb F_q^{(2k-n+1)\times(2k-n)}$, in which the first $2k-n+1$ rows are linearly dependent. This contradicts the fact that its determinant has to be nonzero. Hence, we can conclude that $\rk (G_1)  \geq n-k$.
 \end{proof}

\begin{remark}\label{decoding_window}
For the parameters in consideration, i.e. $k>n-k=\delta$ the generalized Singleton bound of Theorem \ref{thm:singleton} is equal to $2(n-k)+1$, which equals the upper bound for the first column distance $d_1^c$ in Theorem \ref{thm:column}. Hence, an MDP code with these parameters fulfills $d_1^c=d_L^c=d_{free}$. 
More general, each MDP convolutional code with $(n-k)\mid\delta$ is an MDS convolutional code.
\end{remark}

\section{Existence of unit-memory MDP codes over sufficiently large fields}
\label{sec:ex}

As mentioned above, we wish to construct partial unit-memory MDP convolutional codes generated by $G(z) = G_0 + G_1z \in \mathbb{F}_q[z]^{k \times n}$. According to Theorem~\ref{thm:MDP_criterion}, it is sufficient that any $2k \times 2k$ minor of $G_1^c$ as in equation \eqref{eqn:G_1^c}
obtained by choosing at most $k$ columns from the first $n$ columns is nonzero. Let $I \cup J$ denote the columns of $G_1^c$ that we chose in the minor, where $I \subseteq \{1, \dots, n\}$ and $J \subseteq \{n + 1, \dots, 2n\}$. We must have that $0 \leq |I| \leq k \leq |J| \leq n$ and $|I \cup J| = 2k$.
Let us fix $G_0 \in \mathbb{F}_q^{k \times n}$ to be an arbitrary MDS matrix. Furthermore, let $G_1 = (X~\mathbf{0}_{k \times k})$, where $X$ is a $k \times (n - k)$ matrix whose entries we consider as indeterminates $x_{1, 1}, \dots, x_{k, n - k}$. The minor of a matrix is a polynomial in the entries of the matrix, so let us define the \textbf{minor product polynomial} of $G_0$ as
\begin{equation*}
    P(X) = \prod_{\substack{0 \leq i \leq k \\ k \leq j \leq n \\ i + j = 2k}} \prod_{\substack{I \subseteq \{1, \dots, n\} \\ J \subseteq \{n + 1, \dots, 2n\} \\ |I| = i, |J| = j}} P_{I, J}(X),
\end{equation*}
where $P_{I, J}(X) = \det((G_1^c)_{I \cup J})$. Here, $(G_1^c)_S$ denotes the submatrix consisting of the columns of $G_1^c$ indexed by $S \subseteq \{1, \dots, 2n\}$. The matrix $G(z) = G_0 + G_1z$, where $G_1 = (X~\mathbf{0}_{k \times k})$, generates an MDP code if and only if $P(X) \neq 0$. We shall show that if the field size is large enough, then the polynomial $P(X)$ cannot be zero at all points.

\begin{lemma}\label{lem:nonzero}
The polynomial $P(X) \in \mathbb{F}_q[x_{1, 1}, \dots, x_{k, n - k}]$ is not identically zero.
\end{lemma}

\begin{proof}
Let $0 \leq i \leq k$, $k \leq j \leq n$, $i + j = 2k$ and $I \subseteq \{1, \dots, n\}$, $|I| = i$, $J \subseteq \{n + 1, \dots, 2n\}$, $|J| = j$. Denote $M(X) = (G_1^c)_{I \cup J}$. We want to show that the polynomial $P_{I, J}(X) = \det(M(X))$ is not identically zero. We will do this by showing that for some assignment of $X$, the matrix $M(X)$ is invertible.
    
As any $k$ columns of $G_0$ are linearly independent, the $k \times i$ matrix $(G_0)_I$ has rank $i$ so some $i \times i$ submatrix of $(G_0)_I$ is invertible. Without loss of generality, assume that it consists of the bottom $i$ rows of $(G_0)_I$. Denote by $X'$ the matrix of variables that we have chosen with this choice of $J$, e.g., $X'$ is a submatrix of $X$.
\begin{equation*}
    M(X) = 
    \begin{pNiceArray}{ccc|cc|ccc}
        \Block{3-3}{(G_0)_I} & & & \Block{3-2}{X'} & & \Block{3-3}{\mathbf{0}} & & \\
        & & & & & & & \\
        & & & & & & & \\ 
        \hline
        \Block{3-3}{\mathbf{0}} & & & \Block{3-5}{(G_0)_J} & & & \\
        & & & & & & & \\
        & & & & & & &
    \end{pNiceArray}
\end{equation*}
The matrix $X'$ consists of $\geq j - k$ columns and $k$ rows. Set $x'_{1, 1} = x'_{2, 2} = \dots = x'_{j - k, j - k} = 1$ and the other entries to zero. Then, the top right $k \times k$ block of $M(X)$ is zero, so the determinant of $M(X)$ is the product of the determinants of the top left $k \times k$ block and the bottom right $k \times k$ block. The latter block is invertible, since it is a $k \times k$ submatrix of $G_0$. The former is also invertible by construction. This means that $M(X)$ is invertible for this choice of $X$, so $P_{I,J}(X) = \det(M(X)) \neq 0$.
\end{proof}

Let us next compute the degree of the polynomial $P(X)$. For $n \geq 1$ and $0 \leq k \leq n$ define
\begin{equation*}
    d(n, k) = \sum_{\substack{0 \leq i \leq k \\ k \leq j \leq n \\ i + j = 2k}} \binom{n}{i} \binom{n}{j} (j - k).
\end{equation*}

\begin{lemma}\label{lem:homogeneous_degree}
If $G_0 \in \mathbb{F}_q^{k \times n}$, then the polynomial $P(X)$ is homogeneous of degree $d(n, k)$.
\end{lemma}

\begin{proof}
Let $0 \leq i \leq k$ and $k \leq j \leq n$ be such that $i + j = 2k$. We shall show that all the factors $P_{I, J}(X)$ are homogeneous of degree $j - k$. Consider the matrix $M(X) = (G_1^c)_{I \cup J}$ and let $t$ be a variable. Multiply the last $j$ columns of $M(X)$ with $t$ and multiply the last $k$ rows with $t^{-1}$ to obtain the matrix $M'(X)$. It is clear that $M'(X) = M(t X)$. Further, by the multilinearity of the determinant, $\det(M'(X)) = t^{j - k} \det(M(X))$, so
\begin{align*}
    P_{I, J}(t X) &= \det(M(t X)) = \det(M'(X)) \\
    &= t^{j - k} \det(M(X)) = t^{j - k} P_{I, J}(X).
\end{align*}
This means that $P_{I, J}(X)$ is homogeneous of degree $j - k$. The formula for the total degree follows clearly from the definition of $P(X)$.
\end{proof}

The degree of the minor product polynomial follows an expected dual condition given by the following lemma.

\begin{lemma}\label{lem:dual_symmetry_of_degree}
For $n \geq 1$ and $0 \leq k \leq n$, $d(n, k) = d(n, n - k)$.
\end{lemma}

\begin{proof}
The result follows by using the change of variables $i' = n - j$ and $j' = n - i$ in the expansion of $d(n, n - k)$.
\end{proof}

Instead of the total degree of $P(X)$, we may consider the degree of the individual variables. For $n \geq 1$ and $0 \leq k \leq n$ define
\begin{equation*}
    d'(n, k) = \sum_{\substack{0 \leq i \leq k - 1 \\ k + 1 \leq j \leq n \\ i + j = 2k}} \binom{n - 1}{i} \binom{n - 1}{j - 1}.
\end{equation*}

\begin{lemma}\label{lem:individual_degree}
If $G_0 \in \mathbb{F}_q^{k \times n}$ is in standard form, then the degree of any individual variable in $P(X)$ is at most $d'(n, k)$.
\end{lemma}

\begin{proof}
Consider sets $I \subseteq \{1, \dots, n\}$ and $J \subseteq \{n + 1, \dots, 2n\}$ such that $|I| \leq k$ and $|I \cup J| = 2k$. We are interested in computing $\deg_{x_{r, s}}(P_{I, J}(X))$ for some variable $x_{r, s}$, $1 \leq r \leq k$, $1 \leq s \leq n - k$. It is clear that $\deg_{x_{r, s}}(P_{I, J}(X)) \leq 1$, since we never multiply a variable with itself.
\begin{itemize}
    \item If $|I| = |J| = k$, then $\deg_{x_{r, s}}(P_{I, J}(X)) = 0$, since the submatrix formed by these columns is a block upper-triangular with $k \times k$ blocks, so the minor is $\det((G_0)_I) \cdot \det((G_0)_J) \in \mathbb{F}_q^*$.
    \item If $s \notin J$, then $x_{r, s}$ does not appear in $G_{I \cup J}$, so $\deg_{x_{r, s}}(P_{I, J}(X)) = 0$.
    \item If $r \in I$, then the $r$th standard basis vector is one of the columns of $G_{I \cup J}$. The minor corresponding to row $r$ and the column of $x_{r, s}$ will be zero, since the submatrix contains a zero column where the $r$th standard basis vector used to be. Therefore, the coefficient of $x_{r, s}$ is zero in $P_{I, J}(X)$ and $\deg_{x_{r, s}}(P_{I, J}(X)) = 0$.
\end{itemize}
Therefore, if $\deg_{x_{r, s}}(P_{I, J}(X)) = 1$, then $|I| < k$, $|J| > k$, $r \notin I$ and $s \in J$. With these observations, we find that
\begin{equation*}
    \deg_{x_{r, s}}(P_{I, J}(X)) \leq \sum_{\substack{0 \leq i \leq k - 1 \\ k + 1 \leq j \leq n \\ i + j = 2k}} \binom{n - 1}{i} \binom{n - 1}{j - 1} = d'(n, k). 
\end{equation*}
\end{proof}

Again, the individual degree follows a dual condition given by the following lemma.

\begin{lemma}\label{lem:dual_symmetry_of_individual_degree}
For $n \geq 1$ and $0 \leq k \leq n$, $d'(n, k) = d'(n, n - k)$.
\end{lemma}

\begin{proof}
The result follows by using the change of variables $i' = n - j$ and $j' = n - i$ in the expansion of $d'(n, n - k)$.
\end{proof}

The following is a standard result and follows, for example, from the Combinatorial Nullstellensatz.

\begin{lemma}[{\cite[Lemma~2.1]{alon1999combinatorial}}]\label{lem:individual_degree_required_field_size}
Let $P(x) \in \mathbb{F}_q[x_1, \dots, x_n]$ be a nonzero polynomial such that the individual degree of any variable in $P(x)$ is at most $d$. If $q > d$, then there exists $\alpha = (\alpha_1, \dots, \alpha_n) \in \mathbb{F}_q^n$ such that $P(\alpha) \neq 0$.
\end{lemma}

We are now ready to state the following existence result for MDP codes. Notice that the trivial bound $\binom{n}{k} \leq n^k$ gives $d'(n, k) = d'(n, n - k) = \mathcal{O}(n^{2(n - k)})$. 

\begin{theorem}
Let $k > n - k = \delta$ and $G_0 \in \mathbb{F}_q^{k \times n}$ be a generator matrix of an MDS code in standard form. If $q > d'(n, k) = \mathcal{O}(n^{2\delta})$, then there exists a matrix $X$ such that all nontrivial $2k \times 2k$ minors of $G_1^c$ are nonzero, where $G_1 = (X~\mathbf{0}_{k \times k})$. In other words, $G(z) = G_0 + G_1z$ generates an $(n, k, \delta)$ MDP convolutional code over $\mathbb F_q$.
\end{theorem}

\begin{proof}
As $P(x) \not\equiv 0$ and $q$ is larger than the individual degree of any of the variables by Lemma~\ref{lem:individual_degree}, then by Lemma~\ref{lem:individual_degree_required_field_size} there exists a matrix $X \in \mathbb{F}_q^{k \times (n - k)}$ such that $P(X) \neq 0$. This means that all nontrivial $2k \times 2k$ minors of $G_1^c$ are nonzero for this choice of $X$.
\end{proof}

\section{Construction}\label{sec:construct}

In this section, we give an explicit construction for a matrix $G_1^{c}$ given by equation \eqref{eqn:G_1^c} 
such that $G(z)=G_0+G_1z$ is a generator matrix for a partial unit-memory MDP convolutional code. To reduce encoding and decoding complexity, we first fix $G_0\in\mathbb F_{q}^{k \times n}$ to be a structured MDS matrix over a field of smallest possible size. In the next step, we search for $G_1$ of the form $G_1 = (X~\mathbf{0}_{k \times k})$ where $X\in\mathbb F_{q^d}^{k\times (n-k)}$  minimizes the number of nonzero entries which lie in the field extension $\mathbb F_{q^d}$ of $\mathbb{F}_q$ with $d$ as small as possible.

\begin{lemma}\label{lem:MDS}
If $G_1 = (X~\mathbf{0}_{k \times k})$ and the convolutional code generated by $G(z) = G_0 + G_1z$ is MDP, then $(G_0~X)$ has the MDS property.
\end{lemma}

\begin{proof}
Let $m_1$ be any $k\times k$-minor of $(G_0~X)$. If we choose in $G_1^c$ the corresponding $k$ columns together with the last $k$ columns of $G_1^c$, then the corresponding non-trivial minor has the form $m_1\cdot m_2$ where $m_2$ is a minor of $G_0$ (due to the block diagonal structure of the submatrix corresponding to the minor). Hence, according to Theorem \ref{thm:MDP_criterion}, $m_1$ as to be nonzero and hence, $(G_0~X)$ has the MDS property.
\end{proof}

For our construction, we start with a superregular matrix $G_0\in\mathbb F_q^{k \times n}$ and consider the matrix $X \in \mathbb{F}_{q^d}^{k \times (n - k)}$ given by 
\small
\begin{equation}\label{eqn_Matrix_X}
    X=\begin{pmatrix}
        \alpha & 0 & 0 & \cdots & 0 \\
        0 & \alpha^2 & 0 & \cdots & 0 \\
        0 & 0 & \alpha^3 & \cdots & 0 \\
        \vdots & & & \ddots&\vdots \\
        0 & \cdots  & \cdots& 0 & \alpha^{n - k}\\
        0 & 0 & 0 & \cdots & 0 \\
        \vdots & & & & \vdots \\
        0 & 0 & 0 & \cdots & 0 \\
    \end{pmatrix},
\end{equation}
\normalsize
where $\alpha \in \mathbb{F}_{q^d}$ has minimal polynomial over $\mathbb F_q$ of degree $d$.

\begin{theorem}\label{diag}
Let $k>n-k=\delta$ and $d = \left \lceil\frac{\delta^2-1}{4}\right \rceil + 1$. Let $G_0 \in \mathbb F_{q}^{k\times n}$ be a superregular matrix and $G_1 = (X~\mathbf{0}_{k \times k}) \in \mathbb{F}_{q^d}^{k \times n}$ with $X$ given by equation \eqref{eqn_Matrix_X}. Then, all nontrivial minors of
\begin{equation*}
    G_1^c = \begin{pmatrix}
        G_0 & G_1 \\
        \mathbf{0} & G_0
    \end{pmatrix}
\end{equation*}
are nonzero, i.e., we get an $(n,k,\delta)$-MDP code over $\mathbb F_{q^d}$.
\end{theorem}
\begin{proof}

For integers $i, j, \ell$, we consider any $2k \times 2k$ submatrix $M$ of $G_1^c$ by choosing $i$ columns out of first $n$ columns, $j$ columns out of the next $n-k$ and $\ell$ columns out of the last $k$ columns. If $M$ corresponds to a non-trivial minor, then
$$
    i +j +\ell = 2k,\quad \text{ where } \quad 1 \leq i \leq k, \  0 \leq j \leq n-k.
$$
Thus, we can write 
$$
 M= \begin{pmatrix}
    U & V \\
    \mathbf{0} & W\\
\end{pmatrix}\quad\text{with}\quad V=[\hat{V}\ 0_{k\times\ell}], 
$$
where $U$ is a submatrix of $G_0$ of size $k \times i$, $V$ is a submatrix of $G_1$ of size $k \times (j+\ell)$, and $W$ is a submatrix of $G_0$ of size $k \times (j+\ell)$. 
We need to show that $\det(M) \neq 0$.  
When viewed as a polynomial in $\alpha$, we need to consider the highest and lowest exponent of $\alpha$ in $\det(M)$ in comparison to the degree of the minimal polynomial of $\alpha$.
We determine the determinant of $M = (m_{\nu, \mu})_{2k \times 2k}$ using the Leibniz formula given by 
$$
    \det(M) = \sum_{\sigma \in S_{2k}} \mathrm{sgn}(\sigma) \prod_{\mu =1}^{2k} m_{\sigma(\mu), \mu}
$$
We only need to consider $\sigma \in S_{2k}$ such that $m_{\sigma(\mu), \mu} \neq 0$ for all $1 \leq \mu \leq 2k$. 
Since $G_0$ is a superregular matrix, each entry of $U$ and $W$ is nonzero. Let the nonzero entries of $V$ (which are obtained from columns of $X$) be $\alpha^{t_1}, \ldots, \alpha^{t_j}$, where
\begin{equation}\label{eqn:r1<rj}
    1 \leq t_1 < t_2 < \cdots < t_j \leq n-k.
\end{equation}
Let $T:= \{ t_1, \ldots, t_j \} \subseteq [n-k]$ so that $m_{t_v, i+v} = \alpha^{t_v}$ for $1 \leq v \leq j$. 
 For $S, S'\subseteq[2k]$, let $M_{S,S'}$ denote the submatrix of $M$ formed by selecting the rows of the index set $S$ and columns of the index set $S'$. We need to consider $\sigma \in S_{2k}$ such that 
$$
 \sigma(\mu ) \in  \begin{cases}
     [k] &{ \text{ if } 1 \leq \mu \leq i}, \\
     T \cup [2k]\setminus[k] &{ \text{ if }  i+1 \leq \mu \leq i+j}, \\
     [2k]\setminus[k] &{ \text{ if }  i+j+1 \leq \mu \leq 2k}. \\
 \end{cases}
$$
For each such $\sigma$ we have to choose $k$ entries of $M$ in the first $k$ rows, i.e. $k$ entries from $[U\ \ \hat{V}]$. As we also need to choose $i$ entries in the first $i$ columns, i.e. $i$ entries of $U$, we choose exactly $k-i$ entries in $\hat{V}$.
We denote by $R\subset[j]$ the set of indices of the $k-i$ columns of $\hat{V}$ corresponding to this $k-i$ chosen entries and define by $t_R = \{t_r \mid r \in R\}$ the set of corresponding rows.
Therefore, 
\begin{align*}
\det(M) &= \sum_{\sigma \in S_{2k}} \mathrm{sgn}(\sigma) \prod_{\mu =1}^{2k} m_{\sigma(\mu), \mu}\\
&=\sum_{\substack{R\subset [j]\\ |R|=k-i}}\sum_{\substack{\sigma \in S_{2k}\\ \sigma(i+r) = t_r \forall r \in R}}\hspace{-6mm}\mathrm{sgn}(\sigma)\prod_{r\in R}\alpha^{t_r} \hspace{-3mm}\prod_{\mu \in[2k]\setminus\{i + R \}}\hspace{-5mm}m_{\sigma(\mu), \mu}\\
&=\sum_{\substack{R\subset [j]\\ |R|=k-i}}\prod_{r\in R}\alpha^{t_r} \hspace{-3mm} \sum_{\substack{\sigma \in S_{2k}\\ \sigma(i+r)=t_r \forall r\in R\\ \sigma([i])\subset[k]\setminus t_R
}}\hspace{-6mm}\mathrm{sgn}(\sigma) \prod_{\mu \in[2k]\setminus\{i + R\}}\hspace{-5mm}m_{\sigma(\mu), \mu}.
\end{align*}
Hence, $\det(M)$ has the form
$$
\det(M) = \sum_{v=\lambda}^\rho A_v\alpha^v \quad \text{with} \quad \rho=\sum_{r = 1}^{k - i} t_{j + 1 - r},\   \lambda=\sum_{r=1}^{k-i} t_r
$$
To get the coefficient of the lowest term, set $R = [k-i]$, i.e.
$$
    A_\lambda = \sum_{\substack{\sigma \in S_{2k}\\ \sigma(i+r)=t_r,\ r \in [k-i]\\ \sigma([i])\subset[k]\setminus t_{[k-i]}
}}\mathrm{sgn}(\sigma) \prod_{\mu \in[i]\cup ([2k]\setminus[k])} m_{\sigma(\mu), \mu}.
$$

 Recall that $\mathrm{sgn}(\sigma)=(-1)^N$ where $N$ is the number of inversions defined via $1\leq a<b\leq 2k$ with $\sigma(a)>\sigma(b)$,  where $a,b$ are choices of columns.
We split the columns of $M$ into $3$ blocks with column indices 
$$C_1=[i],\quad C_2=R+i=[k]\setminus[i],\quad C_3=[2k]\setminus[k]$$
and denote by $N_{\eta,\beta}$ the number of inversions of $\sigma$ with $a\in C_\eta$ and $b\in C_\beta$ for $1\leq \eta\leq \beta\leq 3$, i.e.
$$N=\sum_{1\leq \eta\leq \beta\leq 3}N_{\eta,\beta}=N_{1,1}+N_{1,2}+N_{3,3}$$
since $N_{\eta,\beta}=0$ for $(\eta,\beta)\in\{(2,2),(1,3),(2,3)\}$.
Also note that $N_{1,2}$ only depends on $R$ and not on $\sigma$ (for $\sigma$ as above). Define 
$$\hat{\sigma}_1:=\sigma|_{[i]}^{},\quad
 \hat{\sigma}_2:=\sigma|_{[2k]\setminus[k]}^{},$$  
  $$\hat{\sigma}_1([i])=\{\gamma_1,\hdots,\gamma_{i}\}\ \text{with}\  \gamma_1<\cdots<\gamma_{i}$$ 
 and define $\sigma_1:=f\circ\hat\sigma_1$ with $f(\gamma_z)=z$ for $z\in[i]$.
  Write 
  $\sigma_2(y):=\hat\sigma_2(y+k)-k$ for $y \in [k]$. 
  Then, 
 \begin{align*}
     (-1)^{N_{1,1}} = \mathrm{sgn}(\sigma_1),\ (-1)^{N_{3,3}} = \mathrm{sgn}(\sigma_2) , \\
 \text{and } \mathrm{sgn}(\sigma) = (-1)^{N_{1,2}} \mathrm{sgn}(\sigma_1)\mathrm{sgn}(\sigma_2).
 \end{align*}

Thus, the coefficient $A_{\lambda}$ of the lowest term in $\det(M)$ is 
\begin{align*}
 & (-1)^{N_{1,2}} \hspace{-1.5mm}\sum_{\substack{\sigma_1\in S_{i}\\ \sigma_2\in S_{k}
}}\hspace{-1mm}\mathrm{sgn}(\sigma_1)\mathrm{sgn}(\sigma_2)\hspace{-0.5mm}\prod_{\mu \in[i] }\hspace{-0.5mm}u_{\hat{\sigma}_1(\mu), \mu}\hspace{-2.5mm}\prod_{\mu \in [2k]\setminus[k]}\hspace{-2.5mm}w_{\hat{\sigma}_2(\mu)-k, \mu-i}\\
&=(-1)^{N_{1,2}} \det(U_{[k]\setminus t_{[k-i]},[i]})\det(W_{[k],[2k-i]\setminus[k-i]})
\end{align*}

Since $U_{[k]\setminus t_{[k-i]},[i]}$ and $W_{[k],[2k-i]\setminus[k-i]}$ are submatrices of the superregular matrix $G_0$, one obtains $A_\lambda  \neq 0$ and thus, $\det(M)$ is not the zero polynomial in $\alpha$. 

It remains to show that
\begin{equation}\label{degree}
  \rho -\lambda =  \sum_{r = 1}^{k - i} t_{j + 1 - r} - \sum_{r=1}^{k-i} t_r
\end{equation}
is smaller than the degree of the minimal polynomial of $\alpha$. 
Using equation  \eqref{eqn:r1<rj}, one gets 
$s \leq t_s \leq n-k-(j-s)$ for $s\in[j]$, i.e.

\begin{align*}
        \left(\sum\limits_{r=1}^{k-i}t_{j+1-r} \right)- 
\left(\sum\limits_{r=1}^{k-i}t_r\right) &= \sum\limits_{r=1}^{k-i} (t_{j+1-r} - t_{r}) \\
&\leq  \sum_{r=1}^{k-i} [(n-k-r + 1)-r]\\
&= \sum\limits_{t=1}^{k-i} [n-k-2t+1] \\&
= (k-i) (n-k+1)-2 \sum\limits_{r=1}^{k-i} r\\&=(k-i)(n-k+1-(k-i+1))
\end{align*}

But $(k-i) (n-k-(k-i))$ attains the maximal value  when $k-i = \lfloor \frac{ \delta }{2} \rfloor$ where $\delta = n-k$ and 
$$
\max_{i \in [k]} [(k-i) (n-k-(k-i))] = \begin{cases}
    \frac{\delta^2}{4} &{ \text{if } \delta \text{ is even},}\\
    \frac{\delta^2-1}{4} &{ \text{if } \delta \text{ is odd}}.\\ 
\end{cases}
$$
Hence equation \eqref{degree} is smaller than $d$, which concludes the proof.
      \end{proof}

\begin{remark}
    Note that $G_1$ has minimum number of non-zero entries, i.e., $n-k$; see Remark \ref{minrank}. Furthermore, we choose $G_0$ to be a Cauchy matrix to minimize encoding/decoding complexity. Note that this requires $q\geq n+k$. 
\end{remark}

\section{Computational complexity}
\label{sec:comp}
Let $G(z) \in \mathbb{F}_{q^d}[z]^{k \times n}$ be a generator matrix of a convolutional code and $u(z) \in \mathbb{F}_{q^d}[z]^k$ a message. The encoding of $u(z)$ is $v(z) = u(z)G(z)$. If the convolutional code is (partial) unit-memory, then the encoding process consists of multiplying a vector over $\mathbb F_{q^d}$ with the matrices $G_0$ and $G_1$, where $G(z) = G_0 + G_1z$. In particular,
\begin{align}\label{eq:enc}
\begin{split}
    u(z)G(z) 
   & = u_0 G_0 + \sum_{j=1}^{\deg(u(z)) + 1} (u_{j-1} G_1 + u_{j} G_0) z^j.\\
\end{split}
\end{align}

In step $j \in [\deg(u)+1]$ of \eqref{eq:enc} we compute $u_{j-1} G_1 + u_{j} G_0$, where $u_{j-1}, u_{j} \in \mathbb{F}_{q^d}^k$.
For simplicity, we simply count the multiplications (as is the convention in algebraic complexity theory \cite{bürgisser1996algebraic}) and give the complexity in terms of $\mathcal{M}(q)$, the complexity of multiplying two elements in the  field $\mathbb{F}_q$. 

\begin{theorem}
    The encoding complexity (\emph{i.e.,} the number of arithmetic operations) of each time step $j$ in \eqref{eq:enc} for the code constructed in Section~\ref{sec:construct} is given by 
    $\mathcal{O}(d [ n \log^2 n  + (n-k)]\mathcal{M}(q))=\mathcal{O}(d  n \log^2 n\mathcal{M}(q))$
    where $\mathcal{M}(q)$ is the complexity of multiplying two elements in the base field $\mathbb{F}_q$.
\end{theorem}

\begin{proof}

For simplicity, we only consider the complexity of encoding if $G_0$ is a Cauchy matrix.
The complexity of multiplying a vector $u \in \mathbb{F}^k_{q^d} $ by a Cauchy matrix with entries in $\mathbb{F}_q^k$ is $\mathcal{O}(d n \log^2 n \mathcal{M}(q) )$ 
where the factor of $\mathcal{O}(d)$ comes from the fact that we can treat a vector in $\mathbb{F}_{q^d}^k$ like a matrix in $\mathbb{F}_q^{d \times k}$ and repeat the algorithm $d$ times (once for each row) and the factor $\mathcal{O}( n \log^2 n \mathcal{M}(q) )$ comes from the famous result in \cite{doi:10.1137/0908017}.

One can directly count that there are at most 
$\mathcal{O}((n-k)d\mathcal{M}(q))$
operations corresponding to the multiplication $yG_1$ using the fact from \cite{chudchud88} that 
$\mathcal{M}(q^d) = \mathcal{O}(d  \mathcal{M}(q))$  
and counting that there are at most $n-k$ operations $y_{i}\alpha^j$ and each such multiplication is in $\mathbb{F}_{q^d}^k$. 
\end{proof}

The constructions in \cite{chen2023lower} and \cite{luo2023construction} use structured matrices $G_0$, $G_1$ where all entries are from $\mathbb F_{q^{\delta}}$ with $q=O(n)$, requiring $O( n \log^2 n )$ multiplications in $\mathbb F_{q^\delta}$ in each encoding step. Considering  $\mathcal{M}(q^\delta) = c_q \delta \mathcal{M}(q)$ (see Section 18.5 of 
\cite{bürgisser1996algebraic}, the original result is due to \cite{chudchud88}), we achieve a smaller encoding complexity than in previous papers if $d\approx\frac{\delta ^ 2}{4} $ is as in our construction and $c_q>\frac{\delta}{4}$.
However, we have discovered via computer search that our construction is MDP with very high probability for $d = \delta $ for small $n,k$; see Table I.
If this can be proven for general $n,k$,
we would achieve a better complexity than \cite{chen2023lower} and \cite{luo2023construction} for all values of $c_q$ (since $c_q \geq 2$; see \cite{bürgisser1996algebraic}).  
Moreover, note that in our construction the complexity of multiplying with $G_1$ is negligible.
We plan to explore this and the decoding complexity as part of future work.

\begin{table}[t]
    \centering
    \begin{tabular}{c|cccccc}
        & $n = 7$ & $n = 8$ & $n = 9$ & $n = 10$ & $n = 11$ & $n = 12$ \\ \hline
        $\delta = 3$ & $1.000$ & $1.000$ & $1.000$ & $1.000$ & $1.000$ & $1.000$\\
        $\delta = 4$ & & & $0.995$ & $0.978$ & $0.976$ & $0.973$ \\
        $\delta = 5$ & & & & & $0.995$ & $0.991$ \\
    \end{tabular}
    \caption{The proportion of constructions as in Theorem 5 that are MDP when $\alpha$ has a minimal polynomial of degree $\delta$ over $\mathbb{F}_q$, where $q \approx n + k$. The test is performed by choosing 1000 random samples.}
    \label{tab:random_entries}
\end{table}

\newpage
\bibliographystyle{unsrt}
\bibliography{itw}

\end{document}